\documentclass{jams-l}

\def\anonymize{}

\usepackage{amsmath, amssymb, amsthm}
\usepackage[dvipsnames]{xcolor}
\usepackage[colorlinks=true,
            linkcolor=BrickRed,
            urlcolor=MidnightBlue,
            citecolor=BrickRed,
            filecolor=black]{hyperref}
\usepackage[capitalize,nameinlink]{cleveref} 

\usepackage{cleveref}
\usepackage{csquotes}
\usepackage{mathtools}

\newtheorem{theorem}{Theorem}
\newtheorem*{theorem*}{Theorem}
\newtheorem*{claim*}{Claim}
\newtheorem{lemma}[theorem]{Lemma}
\newtheorem{claim}[theorem]{Claim}
\newtheorem{corollary}[theorem]{Corollary}
\newtheorem*{definition}{Definition}
\newtheorem*{remark}{Remark}

\newcommand{\rank}{\mathrm{rank}}

\newcommand{\N}{\mathbb{N}}

\newcommand{\R}{\mathbb{R}}

\newcommand{\cost}{\mathsf{space}}

\newcommand{\state}[1]{s^{(#1)}}

\newcommand{\smap}[1]{u^{(#1)}}
\newcommand{\outmap}[1]{m^{(#1)}}

\title{On the space complexity of online convolution}

\ifdefined\anonymize
\author[Andersson]{Joel Daniel Andersson}
\address{Institute of Science and Technology Austria (ISTA)}
\email{joel.andersson@ist.ac.at}

\author[Yehudayoff]{Amir Yehudayoff}
\address{Department of Computer Science, The University of Copenhagen,
and Department of Mathematics, Technion-IIT}
\email{amir.yehudayoff@gmail.com}

\thanks{J.D.A.\ performed this work while at University of Copenhagen, supported by Providentia, a Data Science Distinguished Investigator grant from Novo Nordisk Fonden, and BARC (Villum Investigator grant 54451). A.Y.\ is supported by a DNRF chair grant and BARC}
\else
\author{}
\fi

\date{}

\begin{document}
\begin{abstract}
We study a discrete convolution streaming problem. An input arrives as a stream of numbers $z = (z_0,z_1,z_2,\ldots)$,
and at time $t$ our goal is to output $(Tz)_t$ where $T$ is a lower-triangular Toeplitz matrix. 
We focus on space complexity;
we define a model for studying the memory-size of online continuous algorithms. 
In this model,
algorithms store a buffer of
$\beta(t)$ numbers in order to achieve their goal.

We characterize space complexity using the language of generating functions. 
The matrix $T$ corresponds to a generating function~$G(x)$. 
When $G(x)$ is rational of degree $d$,
it is known that the 
space complexity is at most~$O(d)$.
We prove a corresponding lower bound; the space complexity is at least~$\Omega(d)$.
In addition, for irrational $G(x)$, we prove that the space complexity is infinite. 
We also provide finite-time guarantees.
For example, for the generating function $\frac{1}{\sqrt{1-x}}$ that was studied in various previous works in the context of differentially private continual counting,
we prove a sharp lower bound on the space complexity; at time $t$, it is at least $\Omega(t)$.
\end{abstract}
\maketitle

\section{Introduction}
Streaming algorithms process data as it arrives while maintaining only a limited buffer of past information.
A fundamental problem in this context is the efficient computation of structured linear transformations.
For a linear transformation~$T$,
given the stream $z = (z_0,z_1,\ldots)$, the goal at time $t$ is to compute~$(Tz)_t$.
We focus on the case that $T$ is a lower-triangular Toeplitz matrix.
Such matrices correspond to discrete convolutions and they appear in a variety of applications, including signal processing, time series analysis, and differential privacy (see e.g.~\cite{brent1999stability,CohenS03,dwork_differential_2010,henzinger2024unifying} and references therein).

We suggest the following streaming model (a formal definition follows). 
The state $s^{(t)}$ of the algorithm at time $t$ is a buffer of $\beta(t)$ numbers. 
When the algorithm receives $z_{t+1}$,
it updates the state $(s^{(t)},z_{t+1}) \mapsto s^{(t+1)}$. 
The algorithm is correct at time $t$ if there is a way to recover $(Tz)_t$ from $s^{(t)}$. 
Our objective is to characterize the minimum space (buffer size) required to solve this task for a given $T$. 

The complexity of this problem is naturally linked to the structure of the generating function $G$ associated with $T$.
It is known~\cite{dj2024efficient} that when $G(x)$ is a rational function of degree $d$, the space complexity is at most $O(d)$. Prior to our work, no corresponding lower bound was known.
We close the gap and establish a matching lower bound of $\Omega(d)$.
Furthermore, we analyze the case of irrational $G(x)$ and prove that for this case, the space complexity is infinite. Our results classify the space complexity for infinite streams based on the rationality-degree of~$G(x)$.

What about finite streams?
For a finite stream of length $t$, we prove that certain generating functions, such as $G_{1/2}(x) = 1/\sqrt{1-x}$, require space at least~$\Omega(t)$.
This gives a formal explanation for a phenomenon observed in differentially private (DP) continual counting \cite{dwork_differential_2010, chan_private_2011}, a central primitive in DP machine learning where space efficiency is critical \cite{kairouz_practical_2021, denissov_improved_2022}.
A leading-constant optimal algorithm for the problem~\cite{fichtenberger2023constant,henzinger_almost_2023} relies on solving the streaming task for the corresponding $T=T[G_{1/2}]$.
The apparent difficulty of solving $T[G_{1/2}]$ with small buffer size motivated subsequent work on space-accuracy tradeoffs~\cite{andersson2024smooth,dj2024efficient,binning2025}.
Our results show that this difficulty is inherent in the exact streaming model, and help explain why such tradeoffs are necessary.

\subsection{The model}

For a sequence $a = (a_0,a_1,\ldots) \in \R^\N$ of real numbers,\footnote{In this text $\N = \{0,1,2,\ldots\}$; matrices and vectors are zero-indexed.}
denote by $T = T[a]$ the lower-triangular Toeplitz matrix
\begin{equation*}
    T_{i,j} = 
    \begin{cases} 
        a_{i-j} & \text{if } i \geq j, \\ 
        0 & \text{if } i < j.
    \end{cases}
\end{equation*}
Given the stream $z = (z_0,z_1,z_2,\ldots) \in \R^\N$,
at time $t \in \N$ we want to compute
$$(Tz)_t$$
which depends only on the history $z_0,\ldots,z_t$.
For example, for $a = (1,1,1,\ldots)$,
we get the problem of releasing all prefix-sums of the input, which is also known as \emph{continual counting}. 

This work focuses on the space complexity of this streaming problem.
Because we are working over the real numbers, we suggest the following computational model.
The model makes sense for all matrices $M$ (and not just for lower-triangular Toeplitz matrices).

The model we consider requires exact computation over $\R$. Computations in Euclidean space become more and more relevant as the role of machine learning and optimization grows. Neural networks are one standard model for computation over $\R$, which supports linear operations as well as many other (continuous) activation functions. 
Another standard model for computations over $\R$ is the Blum–Shub–Smale model, which supports algebraic operations as well as boolean decision making (in the form of inequalities).
A third standard model for computations over $\R$ is that of algebraic circuits.

We suggest the following model for online computations of $\R$.
As most online models, our model highlights memory-size as the main computational cost. 
The model supports all continuous operations but does not  support ``bit access'' to the input.
Concrete motivation for studying this model comes from differential privacy (see \Cref{sec:app} for more details). We believe that variants of this model can lead to interesting insights into other types of continuous online algorithms.

\begin{definition}[streaming model]
A streaming algorithm is a sequence of maps $m^{(t)}$ and $u^{(t)}$ as follows.
Let $\beta : \{-1\} \cup \N \to \N$ be a buffer size function with $\beta(-1) = 1$.
Denote the state at time $t$ by
$\state{t} \in \R^{\beta(t)}$.
Initially, the state is $\state{-1} = 0 \in \R$.
Upon receiving $z_{t}$, 
the algorithm updates its internal state:
    \begin{equation*}
        \state{t} = \smap{t}(\state{t-1}, z_{t}),
    \end{equation*}
    where $\smap{t} : \mathbb{R}^{\beta(t-1)} \times \R \to \R^{\beta(t)}$.
The algorithm then outputs 
    \begin{equation*}
m^{(t)}(\state{t})
    \end{equation*}
    where $\outmap{t}: \R^{\beta(t)} \to \R$.
    We assume that the state maps are continuous; for each $t$,
   the map $(z_0,\ldots,z_t) \mapsto \state{t}$ is assumed to be continuous.
    The algorithm solves the streaming problem defined by 
$M\in\mathbb{R}^{\N \times \N}$ if for every $z$ and every $t$,
$$m^{(t)}(\state{t}) = (M z)_t\,.$$
\end{definition}

The model above assumes that the state maps are continuous (the output maps $m^{(t)}$ are not assumed to be continuous). For computations over $\R$, this assumption is natural and allows the use of algebraic operations $+, \times$ as well as any other continuous function like trigonometric functions, max operations and exponential functions. 
Continuous computations, however, are restricted 
and cannot perform ``arbitrary operations''. 

A continuity restriction on the state maps is necessary because without it every problem can be solved with a buffer size of two.
The reason is that we can encode $z = (z_0,\ldots,z_t) \in \R^{t+1}$ using a pair $(r,x) \in \R \times [0,1]$ as follows.
Let $\varphi:[0,1] \to [-1,1]^{t+1}$ be a continuous space-filling curve. For each non-zero $z \in \R^{t+1}$,
let $r = \|z\|_2$
and let $x \in [0,1]$ be so that $\varphi(x) = \frac{z}{r}$ (for $z=0$, define $r=0$ and $x=0$). 
The map $z \mapsto (r,x)$ is not continuous, but
we can recover $z$ from $(r,x)$ using the continuous map $(r,x) \mapsto r \cdot \varphi(x)$.

\begin{remark}
The lower-bound mechanism we develop can be extended to approximate computation as long as the maps $m^{(t)},u^{(t)}$ are assumed to be Lipschitz and when the underlying matrix $($see \Cref{lemma:small_buffer_small_rank}$)$ has many large eigenvalues. 
The lower bounds can also be extended to a stronger model in which, at time $t$, the state update function takes as input all the future part of $z$ $($that is, $(z_{t'})_{t'\geq t})$ and not just $z_t$.
   This stronger model is relevant, for example, for the privacy setting where $z$ is a noise vector that can be sampled as the algorithm sees fit. 
   See \Cref{sec:future} for more details.
\end{remark}

We think of the buffer size $\beta$ as the memory cost of the algorithm.
As is standard in computational complexity theory,
algorithms have costs and problems have complexities.

\begin{definition}
For time $t$ and a time bound $\tau \geq t$,
define $\cost_M(t;\tau)$
as the minimum value $b$ such that there is a
streaming algorithm with buffer size $\beta$
that computes $M$ correctly up to time $\tau$
and $b \geq \max_{t' \leq t} \beta(t')$.
For infinite streams,
the space complexity $\cost_M(t)$ is
defined as $\cost_M(t)=\cost_M(t;\infty)$.
For a sequence $a$ that defines the Toeplitz matrix $T = T[a]$, we use the notation
$\cost_a(t, \tau) \coloneqq \cost_{T[a]}(t, \tau)$.
\end{definition}

Before moving to our lower bounds, we address a couple of basic issues.  
There are two natural ways to decompose the target matrix $M$: we can write $M$ as $M = LR$
and we can write $M$ as $M = A+B$.
The former corresponds to breaking the streaming problem into two problems sequentially, and the latter into two problems in parallel.
The following lemma shows that such a decomposition does not allow for saving on resources.

\begin{lemma}[decompositions do not help] \label{lem:robust-decomposition}
   Let $M$ be a 
   lower-triangular matrix and let $t,I \in \N$.
   If $M = LR$ for lower-triangular matrices $L,R$ then 
   \begin{equation*}
       \cost_L(t;t+I) + \cost_R(t;t+I) \geq \cost_M(t;t+I)\,.
   \end{equation*}
    If $M = A + B$ for lower-triangular matrices $A,B$ then
   \begin{equation*}
       \cost_A(t;t+I) + \cost_B(t;t+I) \geq \cost_M(t;t+I)\,.
   \end{equation*}
\end{lemma}

We also observe the following basic property of algebraic algorithms
(we do not rely on it, but we find it informative). An algorithm is called algebraic if for all $t$, the maps $m^{(t)}$ and $\smap{t}$ are polynomial maps.

\begin{lemma}[algebraic algorithms are linear]
\label{lem:lin}
If the algorithm defined by the maps $\{ (\smap{t}, \outmap{t}) : t\in\N\}$ always outputs a linear function of $z$,
then there is an algorithm that uses only linear maps
that computes the same function of $z$.
\end{lemma}

\subsection{Generating functions}
Let us go back to the setting of Toeplitz matrices. 
It is convenient and useful to think of the underlying sequence $a$ as a generating function.
The sequence $a$ corresponds to the (ordinary) generating function $G = G[a]$ defined by
$$G(x) = \sum_{k=0}^\infty a_k x^k.$$
This expression is just a formal infinite sum.
For more background on generating functions, see e.g.~\cite{lando2003lectures,wilf2005generatingfunctionology}.

Let us consider some examples.
The first example is the continual counting generating function 
$$G_1(x)= \sum_{k=0}^\infty x^k.$$
The space complexity of $G_1$ is $1$
because the algorithm can just update the memory
by adding the observed $z_t$.
A second example is the ``exponential counting'' function
$$G_{e}(x) = \sum_{k=0}^\infty 2^k x^k.$$
The space complexity of $G_e$ is $1$ as well 
(multiply the memory by two,
and then add the observed $z_t$).
The third example is 
$$G_{1/2}(x) = \sum_{k=0}^\infty 2^{-2k} \binom{2k}{k} x^k;$$
motivation for studying this function comes from continual counting in the DP framework~\cite{henzinger_almost_2023,fichtenberger2023constant}.
The fourth example is 
$$G_{C}(x) = \sum_{k=0}^\infty C_k x^k,$$
where $C_k$ is the $k$'th Catalan number. 
The trivial upper bound for the space complexities of both $G_{1/2}$ and $G_C$
is $\cost(t) \leq t+1$.
Is there a clever way to significantly reduce
the buffer size?
It turns out that the property that qualitatively governs the space complexity is rationality.

\begin{definition}
A polynomial $Q(x)$ of degree $d$ is an expression of the form
$Q(x) = \sum_{k=0}^{d} q_k x^k$ where $q_d \neq 0$.
The polynomial $Q(x)$ is invertible if $q_0 \neq 0$.
The generating function $G$ is a rational function of degree at most $d$ if there are polynomials $P,Q$ with $Q$ invertible 
so that $Q(x) \cdot G(x) = P(x)$
where $\deg(Q) \leq d$ and $\deg(P) \leq d-1$. 
We say that $G$ has degree $d$
if it is of degree at most $d$ but not of degree at most $d-1$.
We say that $G$ is irrational if it is not rational (of any degree). 
\end{definition}

Let us go back to the four examples we introduced:
\begin{align*}
 G_1(x) & =  \frac{1}{1-x} , \\
 G_{e}(x) & =  \frac{1}{1-2x} , \\ G_{1/2}(x) & = \frac{1}{\sqrt{1-x}} , \\
 G_{C}(x) & =  \frac{1-\sqrt{1-4x}}{2x} .
\end{align*}
We see that $G_1,G_e$ are degree one rational functions; as simple as they get. 
What about $G_C$ and $G_{1/2}$?

\begin{claim*} Both
$G_{1/2}$ and $G_C$ are irrational. 
\end{claim*}

\subsection{Results}

Our main results are that the degree of rationality
of $G$ determines the space complexity
of the corresponding streaming problem.

\begin{theorem}
\label{thm:ration}
    If $G[a]$ is a rational function of degree $d$, then for $t \geq d$,
    $$\cost_a(t) = d$$
and for $t <d$,
        $$\cost_a(t;t+d-1) = t+1.$$
\end{theorem}

\begin{theorem}
\label{thm:irration}
    If $G[a]$ is irrational, then for all $t \in \N$, 
    $$\cost_a(t) = t+1.$$
\end{theorem}

\Cref{thm:irration} holds just for infinite streams. 
Irrationality is insufficient to derive lower bounds for finite streams because
for every $\tau$,
there exists\footnote{Consider for example $G(x) = G_C(x) + \sum_{k=0}^\tau (1-C_k) x^k$.} an irrational $G[a]$ 
so that $\cost_a(t; \tau) = 1$.
The theorem does not provide any information
on space complexity for a finite time horizon ($\tau < \infty$).

The analysis of finite time horizons  is done by studying the corresponding Hankel matrix (see \Cref{secL:LB}).
Hankel matrices were studied in combinatorics and in 
the analytic theory of continued fractions~\cite{wall2018analytic}.
The properties of Hankel matrices allow to prove that $G_{1/2}$ is essentially as hard as it gets even for finite times.

\begin{theorem}\label{thm:g_one_half_is_hard}
For all $t \in \N$,
$$\cost_{G_{1/2}}(t;2t) = t+1.$$
\end{theorem}

The theorems reveal ``Diophantine-approximation-like'' properties that are related to the space complexity of streaming algorithms.
Diophantine approximation studies ``how well can a real number $\alpha$ be approximated by rationals?'' 
If $\alpha$ is rational, then it has a perfect representation. If $\alpha$ is irrational, it can only be approximated. 
There are standard ways to measure the cost of such approximations.
The ``hardest to approximate'' real numbers turn out to be irrational algebraic numbers (e.g.\ the Thue--Siegel--Roth theorem). For example, the golden ratio $\frac{1+\sqrt{5}}{2}$ has ``high complexity'' in this model (notice the similarity to $G_C)$.
In our context, space complexity  corresponds to ``distance from rationality''.
Rational functions have constant complexity, and irrational ones have infinite complexity.
In addition, like the golden ratio,
the algebraic irrationals
$G_{1/2}$ and $G_C$ are ``as far as possible
from the rationals''. 

\subsection{Applications}
\label{sec:app}
The above results also carry algorithmic implications.
Continual counting is about computing prefix-sums of the stream. In our setting,
it corresponds to $T_1 = T[(1,1,1,\ldots)]$.
\emph{Differentially private (DP) continual counting} \cite{dwork_differential_2010,chan_private_2011} solves the same problem, but also requires the streaming algorithm $\mathcal{A}$ to 
satisfy differential privacy \cite{dp_2006,dwork_algorithmic_2013}.

Throughout this section, we only consider the \emph{approximate} DP version of the problem, which sees the most practical use.
We also only consider the case where an upper bound on the length of the stream is known.\footnote{A technique in \cite{chan_private_2011} allows for removing the assumption of a bounded-length stream at the cost of worse guarantees, but complicates the exposition.}

Most DP algorithms for continual counting are based on adding \emph{correlated noise} to the exact output
and have the form \cite{denissov_improved_2022}:
$$\mathcal{A}(z) = T_1 z + L y ,$$
where $L$ is a lower-triangular matrix and $y$ 
is a random Gaussian vector sampled from $\mathcal{N}(0, \sigma^2 I)$.
Equivalent forms that correspond to other factorizations of $T_1$ (such as $T_1(z + T_1^{-1}L y)$) can be analyzed using \Cref{lem:robust-decomposition}.

The noise parameter $\sigma$ is determined by the privacy demands and by the matrix $L$.
Observe that $\mathcal{A}(z)$ is essentially running two instances of our streaming problem: upon receiving $z_t$, it generates a sample $y_t$ and outputs $(T_1z)_t + (Ly)_t$.
Technically $y_t$ can be sampled at any time $t' < t$, but we show in \Cref{sec:future} that our lower bounds are robust to this detail; there is no gain from sampling $y_t$ earlier.
The space complexity of $T_1$ is $1$, and so $L$ determines the space complexity of $\mathcal{A}(z)$.

The streaming model we consider requires exact computation of $(Mz)_t$.
This matches the usual DP formulation, in which $\sigma$ is calibrated for the intended Gaussian noise distribution $\mathcal{N}(0, \Sigma)$ where $\Sigma = \sigma^{2} LL^{T}$ to mask $T_1(z-z')$ for neighboring inputs $z,z'$.
If $Ly$ is only approximated, effectively released as $F(y)$, then without a stability argument, this may increase the error and/or invalidate the privacy calibration;
otherwise, the guarantees should be stated for $F$ rather than $L$.

Space complexity matters for this problem.
Part of the recent interest in DP continual counting is motivated by private machine learning \cite{kairouz_practical_2021,denissov_improved_2022,mcmahan_federated,choquette-choo_multi-epoch_2022,choquettechoo2023amplified}.
Unconstrained stochastic gradient descent (SGD) updates with constant learning rate can be viewed as computing prefix sums on gradient vectors.
Hence, training a model with DP (protecting the gradients) can be reduced to one instance of DP continual counting per model dimension.
Because the model dimension often dominates the number of weight updates~\cite{dj2024efficient}, a linear space complexity is prohibitive in practice.

The choice of the matrix $L$ has a crucial impact on the DP properties of~$\mathcal{A}$.
Choosing $L=\sqrt{T_1}=T_{1/2}$ leads to near-optimal guarantees~\cite{fichtenberger2023constant, henzinger_almost_2023}.\footnote{Recently, explicit non-Toeplitz $L$ improving lower-order terms have been proposed~\cite{HU25,henzinger2025normalizedsquarerootsharper}. No space-efficient implementations of these matrices have been proposed.}
However, attempts to implement $T_{1/2}$ in sublinear space were unsuccessful.
This drove further research into other matrices $L$ and space-accuracy tradeoffs~\cite{andersson2024smooth,dj2024efficient,binning2025,henzinger2025binnedgroupalgebrafactorization}.
The best known bounds for this trade-off~\cite{dj2024efficient} are achieved for $L=T[\tilde{G}_{1/2}] \approx T_{1/2}$, where $\tilde{G}_{1/2}$ is a rational-function approximation of $G_{1/2}$ implemented with memory cost equal to the degree of $\tilde{G}_{1/2}$.
\Cref{thm:ration} implies that this implementation of $T[\tilde{G}_{1/2}]$ is tight, and \Cref{thm:g_one_half_is_hard} motivates why matrices other than $T_{1/2}$ had to be considered.
In a nutshell, there are no sublinear space algorithms for $T_{1/2}$.
Any implementation of $\mathcal{A}$ with $L=T_{1/2}$ for $n$ steps, restricted to continuous operations, requires space $\Omega(n)$.

Additional motivation from private machine learning led to studying the problem of DP continual counting with momentum and weight decay~\cite{KalLamp24}.
Translating this problem to our setting,
there is a momentum parameter $0\leq \mu < 1$ and a weight decay parameter $0 < \lambda \leq 1$. 
To avoid degenerate cases, we can assume that $0\leq \mu < \lambda \leq 1$.
The streaming problem now corresponds to the generating function
$$G_{\lambda,\mu} = \frac{1}{(1-\lambda x)(1-\mu x)}.$$
In~\cite{KalLamp24}, a DP streaming algorithm of the form 
$$\mathcal{A}_{\lambda, \mu}(z) = T_{\lambda, \mu} z + Ly$$ was presented, achieving asymptotically optimal error for $L= T[\sqrt{G_{\lambda, \mu}}]$.
Our methods yield optimal lower bounds for the space complexity of this task as well. 
\begin{theorem}\label{thm:g_one_half_weighted_is_hard}
For $0\leq \mu < \lambda \leq 1$ and for all $t\in\N$,
$$\cost_{\sqrt{G_{\lambda, \mu}}}(t;2t) \geq t-4.$$
\end{theorem}

There is additional literature on DP continual counting with arbitrary weights \cite{BolotFMNT13,henzinger2024unifying,HU25}.
Most algorithms in this domain are also based on adding correlated noise $Ly$, and whenever $L$ is Toeplitz our techniques and results can be leveraged.

\begin{remark}
In a couple of places below we reuse notation.
For example, $\mu$ denotes the momentum parameter above
and in one of the proofs below it denotes a monomial. The meaning is clear from the context. 
\end{remark}

\section{A lower bound mechanism}
\label{secL:LB}

In this section, we develop a mechanism for proving lower bounds on the space complexity of a sequence $a$ or a Toeplitz matrix $T = T[a]$.
For the lower bound, we shall also employ the Hankel matrix $H = H[a]$ defined by
$$H_{i,j} = a_{i+j}.$$
For $I,J \geq 0$, we denote by
$H^{(I,J)}$ the $I+1 \times J+1$ matrix
that corresponds to positions $0 \leq i \leq I$
and $0 \leq j \leq J$ in $H$.

\begin{lemma}[small buffer, small rank]\label{lemma:small_buffer_small_rank}
For all $t,I \in \N$,
    $$\cost_a(t;t+I) \geq \rank(H^{(I,t)})$$
    where $H = H[a]$.
\end{lemma}

The lemma relates space complexity, which is a ``complicated notion'', to the rank of a matrix, which is a notion that we understand pretty well. 

\begin{remark}
A variant of the lemma holds in a more general framework than the one we consider.
It holds for general lower-triangular matrices and not just for lower-triangular Toeplitz matrices.
For every matrix $A$, the buffer size can be bounded from below by the rank of some submatrix of $A$.
\end{remark}

The lower bound is based on a basic topological property;
the Borsuk--Ulam theorem. Denote by $\mathbb{S}^{r-1}$
the unit sphere in $\R^r$:
$$\mathbb{S}^{r-1} = \{x \in \R^r : \|x\|_2 = 1\}.$$

\begin{theorem*}[Borsuk--Ulam]
If $f : \mathbb{S}^{r-1} \to \R^{r-1}$ is continuous
then there is $x \in \mathbb{S}^{r-1}$
so that $f(x) = f(-x)$.
\end{theorem*}

\begin{proof}[Proof of \Cref{lemma:small_buffer_small_rank}]
Consider the restricted setting where the incoming stream $z$ is so that $z_{t'} = 0$ for all $t' > t$, and we are only interested in the outputs $(Tz)_{t'}$
for $t \leq t' \leq t+I$.
Accordingly, for the rest of the proof, assume that $z$
is zero after position $t$,
or equivalently that $z \in \R^{t+1}$.
Denote by $T'$ the first $t+1$ columns of $T$
so that $Tz = T'z$.

Denote by $f(z) \in \R^{\beta(t)}$
the state of the memory at time $t$
when the input is $z \in \R^{t+1}$.
The map $f$ is continuous by assumption on the algorithm.
Denote by $M$ the submatrix of $T'$ formed from row indices $t, t+1,\ldots,t+I$;
it defines a linear map from $\R^{t+1}$ to $\R^{I+1}$.
By definition of rank,
there is a vector space $V \subset \R^{t+1}$
of dimension $r = \rank(M)$
so that $Mv \neq 0$ for every non-zero $v \in V$.
The space $V$ is isomorphic to $\R^r$.
Let $S$ be a copy of the sphere $\mathbb{S}^{r-1}$ in $V$. For every $v \in S$ we have $Mv \neq 0$.

Denote by $g(z) \in \R^{I+1}$ the output of the algorithm
at times $t,t+1,\ldots,t+I$ on input $z$.
By the structure of the algorithm,
$g$ can be factored as $g = h \circ f$
where $h : \R^{\beta(t)} \to \R^{I+1}$
(we do not know that $h$ is continuous).
Assume towards a contradiction that
$\beta(t) < r$.
By the Borsuk--Ulam theorem, there is a $v \in S$ so that $f(v) = f(-v)$
and consequently $g(v) = g(-v)$.
By the correctness criterion, 
$$Mv = g(v) = g(-v) = -Mv$$
so $Mv = 0$ which is a contradiction. 

Because $H^{(I, t)}$ is a column permutation of $M$, the proof is complete. 
\end{proof}

\section{Asymptotic complexity}

In this section, we show that the space complexity of infinite streams is characterized by the rational function degree of the generating function. 

\subsection*{Compactness}

\begin{lemma}\label{lem:compactness}
Let $M$ be a $d \times \N$ real matrix with row rank $d$.
For each $j \in \N$, let $M[j]$ be the $d \times j+1$ matrix of the first $j+1$ columns of $M$.
Then, there is some $j \in \N$ so that
$M[j]$ has rank $d$.
\end{lemma}

\begin{proof}
Assume towards a contradiction that
for all $j$, the matrix $M[j]$ is singular.
That is, there is a unit vector $v[j] \in \R^d$ so that $v[j] M[j] = 0$;
in words, it is in the left kernel of $M[j]$.
The unit sphere is compact so by going to a subsequence we can assume that $v[j]$ tends to a limit $v_* \neq 0$.

Fix $j_0 \in \N$, and
denote by $m \in \R^d$ the $j_0$'th column in $M$.
For all $j > j_0$, we have
$\langle v[j], m \rangle =0$.
As $j \to \infty$, we have $\langle v_*, m \rangle
= \langle v_*-v[j], m  \rangle \to 0$.
So, $\langle v_*, m \rangle = 0$.

It follows that $v_*$ is in the left kernel of $M$, which is a contradiction.
\end{proof}

\subsection*{Hankel matrices}

We recall the following theorem proved by Kronecker; see e.g.~\cite{sarason2004hankel,al2017hankel} and references therein. 

\begin{theorem}
\label{thm:Kron}
Let $a \in \R^\N$.
The generating function $G=G[a]$ is rational of degree $d$
iff $H = H[a]$ has rank $d$ and $H^{(d-1,d-1)}$ has rank $d$. The function $G$ is irrational iff the rank of $H$ is infinite. 
\end{theorem}

\begin{claim}
\label{clm:HankelMat}
Let $a \in \R^\N$
and let $H = H[a]$.
If the rank of $H$ is $d$,
then for every $t \leq d$, the first $t$ columns in $H$ are linearly independent.
\end{claim}

\begin{proof}
If the first $t$ columns of $H$ are linearly dependent, then some column $i \leq t$ is a linear combination of columns $<i$. Because $H$ is an infinite Hankel matrix, all columns after column $i$ are also spanned by columns $<i$ (with the same coefficients).
\end{proof}

\subsection*{Rational functions}

\begin{proof}[Proof of \Cref{thm:ration}]
   Assume that $G=G[a]$ is a rational function of degree $d$
and let $H = H[a]$ be its Hankel matrix.
It is known that because $G[a]$
is rational of degree $d$, 
it holds that $\cost_a(t) \leq d$ for all $t$; see (1.8) and Lemma~3.2 in~\cite{dj2024efficient}.
Their algorithm as stated requires $\beta(0) = d$, but it is possible to implement it with $\beta(t) = \min \{t+1, d\}$.
For $t < d$, let $\state{t}$ be a copy of the input stream $\state{t} = (z_0,\ldots,z_t)$ with the obvious output map $\outmap{t}$.
For $t \geq d$, let $\state{d}$ be the corresponding state from~\cite{dj2024efficient} and use their (fixed) maps $\{(\smap{t}, \outmap{t}) : t\geq d \}$ for all future computations.
It follows that $\cost_a(t) \leq \min\{t+1, d\}$.

Now, let $t<d$.
By \Cref{thm:Kron}, the matrix $H^{(d-1, d-1)}$ has full rank. 
In particular, $\rank(H^{(d-1, t)}) = t+1$.
Invoking \Cref{lemma:small_buffer_small_rank}, we get $\cost_a(t;t+d-1) \geq t+1$.
\end{proof}

\subsection*{Irrational functions}

\begin{proof}[Proof of \Cref{thm:irration}]
Because $G[a]$ is irrational,
the Hankel matrix $H[a]$ has infinite rank.
By \Cref{clm:HankelMat}, for all $t \in\N$, the first $t+1$ columns in $H$ are linearly independent.
Fix $t$ and let $M$ be the first $t+1$ columns of $H$.
By compactness (\Cref{lem:compactness}), 
there is some $j$ so that $H^{(j,t)}$
has rank $t+1$.
By \Cref{lemma:small_buffer_small_rank},
we know that $\cost_a(t) \geq t+1$.
The upper bound $\cost_a(t) \leq t+1$ is trivial.
\end{proof}

\section{Finite time horizons}

Asymptotic space complexity can be understood via degrees as rational functions.
The finite time horizon case for irrational functions can be understood via deeper properties of the corresponding Hankel matrices.

\subsection*{Preliminaries}

In this section, we explain some methods for comparing the ranks of different Hankel matrices. 
We shall use the following simple properties.
If $f(x),g(x)$ are generating functions
and $T[f],T[g]$ are the corresponding Toeplitz matrices then
$$T[f \cdot g] = T[f] \cdot T[g]$$
and
$$T[f + g] = T[f] + T[g].$$

\begin{lemma}
\label{lem:comp}
Let $f,g$ be generating functions
and let $\alpha,\beta,\gamma \in \R$ be
so that $f_0 \neq \alpha$, so that
$g_0 \neq 0$
and so that
$$f g = \alpha g + \beta xg + \gamma.$$
Then, for every $d \in \N$,
the rank of $(H[f])^{(d,d)}$ is at least the rank of $(H[g])^{(d,d)}$ minus three.
In addition, when $\beta = 0$,
the rank of $(H[f])^{(d,d)}$ is at least the rank of $(H[g])^{(d,d)}$ minus two.
\end{lemma}

\begin{proof}
Fix $d$, and denote by $T'[h]$
the principal $2d \times 2d$ submatrix of $T[h]$.
Because the ``upper half'' of $T[h]$ is zero,
$$T'[f \cdot g] = T'[f] \cdot T'[g].$$
Therefore, by assumption,
$$T'[f -\alpha  -\beta x] \cdot T'[g]
=  T'[\gamma].$$
The matrix $T'[h]$ can be partitioned into four equal blocks of size $d \times d$ as
\begin{align*}T'[h] = \left[
\begin{array}{cc}
T_{11}[h] & T_{12}[h]  \\
T_{21}[h] & T_{22}[h]  \\
\end{array}\right].
\end{align*}
By looking at the lower-left block, $T_{21}$,
we see that
\begin{align*}
T_{21}[f -\alpha  -\beta x] \cdot T_{11}[g]
+ T_{22} [f -\alpha  -\beta x] \cdot T_{21}[g]
= T_{21}[\gamma] =0 ,
\end{align*}
because $T_{21}[\gamma]$ is zero.
The matrices $T_{11}[g]$ and $T_{22} [f -\alpha  -\beta x]$ have full rank because $f_0 \neq \alpha$ and $g_0 \neq 0$.
The matrix 
$$M\coloneqq- T_{22} [f -\alpha  -\beta x] \cdot T_{21}[g] \cdot (T_{11}[g])^{-1}$$
has the same rank as $T_{21}[g]$
and
\begin{align*}
T_{21}[f -\alpha  -\beta x]   = M  .
\end{align*}
So,
$$T_{21}[f] 
= M  + T_{21}[\alpha +\beta x]
= M  + T_{21}[\beta x].$$
The matrix $T_{21}[\beta x]$ has at most one non-zero entry
(so its rank is at most one).
The matrix $T_{21}[h]$ is obtained from
$(H[h])^{(d,d)}$ by deleting the first row and last column, and then reversing the columns.
The rank of $(H[f])^{(d,d)}$ is hence at least the rank of $T_{21}[f]$.
The rank of $M$, which is equal to the rank of $T_{21}[g]$, is at least the rank of
$(H[g])^{(d,d)}$ minus two.
\end{proof}

\subsection*{Examples}
Let us start with the Catalan generating function $G_C$.
Denote by $H_C$ the Hankel matrix
that corresponds to the Catalan sequence.
We note the following folklore lemma
(see e.g.~\cite{Aigner1999}).

\begin{lemma}\label{lemma:catalan-full-rank}
For every $d \in \N$,
$$\det(H_C^{(d,d)})=1.$$
In particular, $H_C^{(d,d)}$ has full rank. 
\end{lemma}

Denote by $H_{1/2}$ the Hankel matrix
that corresponds to $G_{1/2}$.
An $n \times n$ matrix has co-rank $r$
if its rank is $n-r$.
The full rank of $(H_C)^{(d,d)}$ can be used with \Cref{lem:comp} to show that $(H_{1/2})^{(d,d)}$ has a co-rank of at most four.
The rank of $(H_{1/2})^{(d,d)}$ is, in fact, known to be full.
\begin{lemma}\label{lemma:square-root}
    For every $d \in \N$, $(H_{1/2})^{(d,d)}$ has full rank. 
\end{lemma}
\begin{proof}
    Proposition~8 in \cite{Aigner1999} states that $\det(H[1/\sqrt{1-4x}]^{(d,d)}) = 2^d$ implying the full rank of $(H[1/\sqrt{1-4x}])^{(d,d)}$.
    It follows that $(H[1/\sqrt{1-x}])^{(d,d)}$ also has full rank.
\end{proof}
\begin{proof}[Proof of~\Cref{thm:g_one_half_is_hard}]
   The lower bound follows from \Cref{lemma:square-root} and \Cref{lemma:small_buffer_small_rank}.
   The upper bound $\cost_a(t;2t)\leq t+1$ is trivial.
\end{proof}

Finally, recall the generating function
$$G_{\lambda,\mu} = \frac{1}{(1-\lambda x)(1-\mu x)}$$
for $0\leq \mu < \lambda \leq 1$, and its square root $\sqrt{G_{\lambda, \mu}}$.
Our main lemma here is the following.

\begin{lemma}\label{lemma:ab_gf_full_rank}
Let $b \neq c$ be non-zero complex numbers
and let
\begin{equation*}
        G(x) = \frac{8}{(b-c)^2} \frac{1 - \frac{b + c}{2}x - \sqrt{(1-b x)(1- c x)}}{x^2} .
    \end{equation*}
    For all $d \in \N$, the matrix $(H[G])^{(d,d)}$ has full rank.
\end{lemma}

Our argument relies on the following theorem by Junod (Theorem 2 in \cite{Junod2003}).

\begin{theorem}
\label{thm:determinant_junod}
    Consider a generating function 
    $$G(x) = \frac{1}{1 - W(x)}$$ with $W(0) = 0$. Suppose that $w(x) = \frac{W(x)}{x} - W'(0)$ satisfies $$w(x) = x(\alpha + \beta w(x) + \gamma w(x)^2)$$ for some parameters $\alpha\neq 0, \beta$ and $\gamma$ in $\mathbb{C}$.
    Then for any $d\in\N$:
    \begin{equation*}
        \det((H[G])^{(d, d)}) = \alpha^{d(d+1)/2}\gamma^{d(d-1)/2}\enspace.
    \end{equation*}
\end{theorem}

Junod also remarks that $w(x)$ can be written as
\begin{align*}
w(x) = \frac{F(x)-1}{x F(x) } - F'(0) 
\end{align*}
and as
\begin{equation*}
    w(x) = \frac{1 - \beta x - \sqrt{(1-\beta x)^2 - 4\gamma\alpha x^2}}{2\gamma x},
\end{equation*}
when $\gamma \neq 0$.

\begin{proof}[Proof of \Cref{lemma:ab_gf_full_rank}]
    To invoke \Cref{thm:determinant_junod}, we will first prove that $W(x) = 1 - 1/G(x)$ satisfies $W(0) = 0$.
    We proceed:
    \begin{align*}
        W(x) &= 1 - \frac{(b-c)^2}{8} \frac{x^2}{1-\frac{b+c}{2}x - \sqrt{(1-bx)(1-cx)}}.
        \end{align*}
For some generating functions $R_1,R_2,R_3$,
we can write
\begin{align*}
&\sqrt{(1-bx)(1-cx)} \\
& =  \sqrt{1 - (b+c)x + bc x^2}  \\
& = 1 - \frac{(b+c)x}{2} + \frac{bc x^2}{2}
- \frac{x^2(b+c-bc x)^2}{8} -  \frac{x^3(b+c-bc x)^3}{16} + x^4 R_1 \\
& = 1 - \frac{(b+c)x}{2} +x^2 \Big( \frac{bc }{2} - \frac{(b+c-bc x)^2}{8} - 
 \frac{x(b+c-bc x)^3}{16} + x^2 R_1 \Big) \\
& = 1 - \frac{(b+c)x}{2} + x^2 \Big( \frac{bc }{2}
- \frac{b^2 + c^2 + 2bc}{8}  
+ \frac{2bc(b+c)}{8} x - \frac{(b+c)^3}{16}x+
x^2 R_2 \Big) \\
& =1 - \frac{(b+c)x}{2} - \frac{x^2 (b-c)^2}{8} \Big( 1 + \frac{b+c}{2} x +
x^2 R_3 \Big) .
\end{align*}
So,
   \begin{align*}
        W(x) 
       & = 1- \frac{1}{ 1 + \frac{b+c}{2} x +
x^2 R_3},
        \end{align*}
and we can conclude that $W(0)=0$.
It also follows that $G'(0) = \frac{b+c}{2}$ and by the remark above we can conclude
\begin{align*}
w(x) 
& = \frac{G(x)-1}{x G(x) } - G'(0) \\
& = \frac{\frac{8}{(b-c)^2} \frac{1 - \frac{b + c}{2}x - \sqrt{(1-b x)(1- c x)}}{x^2}-1}{\frac{8}{(b-c)^2} \frac{1 - \frac{b + c}{2}x - \sqrt{(1-b x)(1- c x)}}{x}} - \frac{b+c}{2} \\
& = \frac{8(1 - \frac{b + c}{2}x - \sqrt{(1-b x)(1- c x)})- (b-c)^2x^2}{8x(1 - \frac{b + c}{2}x - \sqrt{(1-b x)(1- c x)})} - \frac{b+c}{2} .
\end{align*}
We claim that
\begin{align*}
 &\frac{8(1 - \frac{b + c}{2}x - \sqrt{(1-b x)(1- c x)})- (b-c)^2x^2}{8x(1 - \frac{b + c}{2}x - \sqrt{(1-b x)(1- c x)})} \\
& \qquad = \frac{x - \frac{x}{2} (1 - \frac{b + c}{2}x + \sqrt{(1-b x)(1- c x)})}{ x^2}.
\end{align*}
Indeed, this is equivalent to 
\begin{align*}
 &x^2 \Big( 8 \Big(1 - \frac{b + c}{2}x - \sqrt{(1-b x)(1- c x)} \Big)- (b-c)^2x^2\Big)  \\
& \qquad = \Big( 8x(1 - \frac{b + c}{2}x - \sqrt{(1-b x)(1- c x)}) \Big) \\
& \qquad \qquad \cdot \Big( x -\frac{x}{2} (1 - \frac{b + c}{2}x + \sqrt{(1-b x)(1- c x)})\Big) 
\end{align*}
which is correct because
\begin{align*}
& \Big( 8x(1 - \frac{b + c}{2}x - \sqrt{(1-b x)(1- c x)}) \Big)  \frac{x}{2} \Big(1 - \frac{b + c}{2}x + \sqrt{(1-b x)(1- c x)}\Big) \\
& =  4x^2 \Big(1 - \frac{b + c}{2}x - \sqrt{(1-b x)(1- c x)}\Big) \Big( 1 - \frac{b + c}{2}x + \sqrt{(1-b x)(1- c x)}\Big) \\
& =  4x^2 \Big( \Big( 1 - \frac{b + c}{2}x \Big)^2 - (1-b x)(1- c x)\Big) \\
& =  4x^2 \Big( \Big(\frac{b + c}{2}x \Big)^2 - bc x^2) \Big) \\
& =  4x^2 \Big( 
\frac{b^2 + c^2 + 2bc}{4}x^2 - bc x^2)\Big) \\
& =  x^2 (b-c)^2 x^2  .
\end{align*}
So,
\begin{align*}
w(x) 
& = \frac{x - \frac{x}{2} (1 - \frac{b + c}{2}x + \sqrt{(1-b x)(1- c x)})}{ x^2}
 - \frac{b+c}{2} \\
& = \frac{1 - \frac{1}{2} (1 - \frac{b + c}{2}x + \sqrt{(1-b x)(1- c x)}) - \frac{(b+c)}{2}x}{ x} \\
& = \frac{1 - \frac{b + c}{2}x - \sqrt{(1-b x)(1- c x)}) }{ 2x} .
\end{align*}
As in the remark above,
\begin{equation*}
    w(x) = \frac{1 - \beta x - \sqrt{(1-\beta x)^2 - 4\gamma\alpha x^2}}{2\gamma x} ,
\end{equation*}
with $\gamma =1$,
with $\beta = \frac{b+c}{2}$
and with $\alpha = \frac{(b-c)^2}{16}$.
It follows that the determinant of $(H[F])^{(d, d)}$ is non-zero.
\end{proof}

Denote by $H_{\lambda,\mu}$ the Hankel matrix that corresponds to $\sqrt{G_{\lambda,\mu}}$.

\begin{corollary}\label{cor:harder-sqrt}
For all $0 \leq \mu < \lambda \leq 1$
and for all $d \in \N$, the matrix $(H_{\lambda,\mu})^{(d,d)}$ has co-rank at most five.
\end{corollary}

\begin{proof}
    Start with $f(x) = \sqrt{(1-\lambda x)(1-\mu x)}$
    and with
    \begin{equation*}
        G(x) = \frac{1 - \frac{\lambda + \mu}{2}x - f}{x^2} .
    \end{equation*}
Because
    \begin{align*}
       f\cdot G &= \frac{(1-\frac{\lambda+ \mu}{2}x) f - (1-\lambda x)(1-\mu x)}{x^2}\\
       &= -G - \lambda \mu + \frac{(\lambda+\mu)(1-f)}{2x}\\
       &= -G - \lambda \mu + \frac{(\lambda+\mu)(1-\frac{\lambda+\mu}{2}x-f)}{2x} + \frac{(\lambda+\mu) \frac{\lambda+\mu}{2} x}{2x} \\
       &= \Big(\frac{\lambda+\mu}{2}x - 1\Big) G + \frac{(\lambda-\mu)^2}{4}.
    \end{align*}
    \Cref{lemma:ab_gf_full_rank} and \Cref{lem:comp} imply that $(H[f])^{(d,d)}$ has co-rank at most three.
    \Cref{lem:comp} implies that $(H_{\lambda,\mu})^{(d,d)}$ has co-rank at most five.
\end{proof}

\begin{proof}[Proof of \Cref{thm:g_one_half_weighted_is_hard}]
   The theorem follows from \Cref{cor:harder-sqrt} and \Cref{lemma:small_buffer_small_rank}.
\end{proof}

\section{Basic properties}

\subsection{Algebraic decompositions do not help}

\begin{proof}[Proof of \Cref{lem:robust-decomposition}]
    We prove the statement for factorizations---the statement for sums follows similarly.
    The observation is that algorithms for solving the streaming problem on $L$ and $R$ yield an algorithm for solving $M$.
    Assume we are given the sequences of maps $u^{(t)}_L, m^{(t)}_L$ for $L$ and $u^{(t)}_R, m^{(t)}_R$ for $R$. On input $z$, denote the corresponding states by $s^{(t)}_L(z) \in \mathbb{R}^{\beta_L(t)}$ and $s^{(t)}_R(z) \in \mathbb{R}^{\beta_R(t)}$.
    By the correctness for $R$, we know
    $m^{(t)}_R\big(s^{(t)}_R\big) = (R z)_{t}$.
    We get an algorithm for $M$ via
    \begin{equation*}
        s^{(t)}_M(z) = (s^{(t)}_L(Rz),s^{(t)}_R(z))\in \mathbb{R}^{\beta_L(t) + \beta_R(t)} 
    \end{equation*}    
where
    \begin{equation*}
         u_M^{(t)}(s^{(t-1)}_M, z_t) \coloneqq \left(u^{(t)}_L\Big(s^{(t-1)}_L, m^{(t)}_R\big(s^{(t)}_R(z)\big)\Big), s^{(t)}_R(z) \Big)\right)
    \end{equation*}
    and
    \begin{equation*}
        m^{(t)}_M(s^{(t)}_M) = m^{(t)}_L(s^{(t)}_L(Rz)) .
    \end{equation*}
    The only property of the algorithm that warrants checking is the continuity of $s^{(t)}_M$.  This follows from the linearity of $m^{(t)}_R\big(s^{(t)}_R(z)\big) = (Rz)_{t}$.
\end{proof}
\subsection{Algebraic algorithms are linear}

An algebraic streaming algorithm is defined by a sequence
of polynomial maps $\smap{t}$ and $\outmap{t}$.
We now explain why we can assume that in this case, the maps $\smap{t}$ and $\outmap{t}$ are linear.
For a multivariate polynomial $p$ and $d \in \N$, denote by $\mathop{\mathcal{H}_d}(p)$ the homogeneous part of degree $d$ in~$p$
(it is a linear functional).

\begin{lemma}\label{lem:lin_pol}
Let $f: \R^n \to \R^k$ and $g : \R^k \to \R$ be two polynomial maps. 
Write $f = (f_1,\ldots,f_k)$,
where each $f_i : \R^n \to \R$ is a polynomial map. 
Then, there are linear maps $L_f: \R^n \to \R^k$ and $L_g : \R^k \to \R$ so that
$$L_g \circ L_f = \mathop{\mathcal{H}_1}(g \circ f).$$
In addition, $(L_f)_i = \mathop{\mathcal{H}_1}(f_i)$ for all $i \in [k]$.
\end{lemma}

\begin{proof}
For polynomials over $\R^n$,
equality as functions is equivalent to equality as formal polynomials. 
Define the affine map $A_f : \R^n \to \R^k$ by
$$(A_f)_i = \mathop{\mathcal{H}_0}(f_i)+\mathop{\mathcal{H}_1}(f_i).$$
Every monomial in $f_i(x) - (A_f)_i(x)$ has degree at least two.
So, for every $k$-variate monomial $\mu$, it holds that
\begin{align*}
\mathop{\mathcal{H}_1}(\mu(f(x))) 
= \mathop{\mathcal{H}_1}(\mu(f(x)-A_f(x) + A_f(x)))
= \mathop{\mathcal{H}_1}(\mu(A_f(x))) .
\end{align*}
It follows that
$$\mathop{\mathcal{H}_1}(g(f(x))) = \mathop{\mathcal{H}_1}(g(A_f(x))).$$
For every $k$-variate monomial $\mu$, there is a polynomial $q = q_{f,\mu}$ so that
$$\mathop{\mathcal{H}_1}(\mu(A_f(x))) = \mathop{\mathcal{H}_1}(q(L_f(x)))$$
where $(L_f)_i\coloneqq\mathop{\mathcal{H}_1}(f_i)$.
By linearity, there is a polynomial $r = r_{f,g}$ so that 
$$\mathop{\mathcal{H}_1}(g(A_f(x))) = \mathop{\mathcal{H}_1}(r(L_f(x))).$$
If $\mu$ is a monomial of (total) degree larger than one, then $\mu(L_f(x))$ has degree larger than one (unless it is zero).
It follows that
$$\mathop{\mathcal{H}_1}(r(L_f(x))) = L_g(L_f(x))$$
where $L_g = \mathop{\mathcal{H}_1}(r)$.
\end{proof}

\begin{proof}[Proof of \Cref{lem:lin}]
The update of the state at time $t$ is
    \begin{equation*}
        \state{t} = \smap{t}(\state{t-1}, z_{t}).
    \end{equation*}
Let $L_u^{(t)}$ be the linear map obtained with $g=\smap{t}$ in \Cref{lem:lin_pol}.
Inductively define states $\sigma^{(t)}$ by $\sigma^{(-1)} = 0$
and    
\begin{equation*}
        \sigma^{(t)} = L_u^{(t)}(\sigma^{(t-1)}, z_{t}).
    \end{equation*}
It follows by induction and \Cref{lem:lin_pol} that
\begin{align}
\label{eqn:lambda}
\mathop{\mathcal{H}_1}(s^{(t)})
= L_u^{(t)}(\mathop{\mathcal{H}_1}(s^{(t-1)}), z_{t})
= L_u^{(t)}(\sigma^{(t-1)}, z_{t})
= \sigma^{(t)}.
\end{align}
By assumption,
the output at time $t$ is a linear function:
    \begin{equation*}
m^{(t)}(\state{t}) = \mathop{\mathcal{H}_1}(m^{(t)}(\state{t})) .
    \end{equation*}
By \Cref{lem:lin_pol} with $g = m^{(t)}$
and using~\eqref{eqn:lambda}, there is a linear map $L_m^{(t)}$ so that  
    \begin{equation*}
\mathop{\mathcal{H}_1}(m^{(t)}(\state{t})) = L_m^{(t)}(\sigma^{(t)}) . \qedhere
    \end{equation*}
\end{proof}

\subsection{Future inputs do not help.}
\label{sec:future}

Our streaming model can be extended to support reading future inputs, which we do next.
As explained in the introduction, this extension is particularly relevant to the differential privacy setting. 
Throughout this section, given $z\in\R^{\N}$, we use the notation $z_{\geq t} = (z_{t'})_{t'\geq t}$ for the input stream with past truncated at $t\in\N$.

\begin{definition}[streaming model with look-ahead]
   The streaming model with look-ahead allows for the algorithm at time $t$ to also query entries in $z_{> t}$ and branch on their value.
   The resulting maps are represented as
   \begin{equation*}
       \hat{u}^{(t)} : \mathbb{R}^{\beta(t-1)} \times \mathbb{R}^{\mathbb{N}} \to \mathbb{R}^{\beta(t)},
       \quad\hat{m}^{(t)} : \mathbb{R}^{\beta(t)} \times \mathbb{R}^{\mathbb{N}} \to \mathbb{R},
   \end{equation*}
   where the state is updated via $\hat{s}^{(t)} = \hat{u}^{(t)}(\hat{s}^{(t-1)}, z_{\geq t})$ where $\hat{s}^{(-1)} = 0$ and the algorithm outputs $\hat{m}^{(t)}(\hat{s}^{(t)}, z_{>t})$.
The map $\hat{s}^{(t)}$ is assumed continuous as a function in $z_{\leq t}$
(that is, for every fixed $z_{>t}$, the map
$z_{\leq t} \mapsto \hat{s}^{(t)}(z)$ is continuous).
\end{definition}
Denote the corresponding space complexity for this look-ahead model by $\widehat{\cost}_M(t; I)$.
\begin{lemma}[the future is useless]\label{lem:future-no-help}
    Let $M\in\mathbb{R}^{\N \times \N}$ be lower-triangular.
    For all $t, I \geq 0$,
    \begin{equation*}
        \widehat{\cost}_M(t; I) = \cost_M(t; I).
    \end{equation*}
\end{lemma}

The lemma shows that all our lower bounds  extend to this look-ahead model.

\begin{proof}
    Let $\hat{\mathcal{A}} = \{ (\hat{u}^{(t)}, \hat{m}^{(t)}) : t\in\mathbb{N}\}$ define a streaming algorithm for $M$ in the look-ahead streaming model.
    Let $z\in\mathbb{R}^{\mathbb{N}}$, and define $z^{\langle t \rangle}$ as $z$ with $z_{> t}$ set to zero.
    Consider another (non-look-ahead) streaming algorithm $\mathcal{A} = \{ (u^{(t)}, m^{(t)}) : t\in\mathbb{N}\}$ defined via
    \begin{equation*}
        u^{(t)}(s^{(t-1)}, z_t) = \hat{u}^{(t)}(s^{(t-1)}, z^{\langle t \rangle}_{\geq t}),
        \qquad m^{(t)}(s^{(t)}) = \hat{m}^{(t)}(s^{(t)}, z^{\langle t \rangle}_{> t}),
    \end{equation*}
    where both maps are taken from $\hat{A}$ but with the \enquote{future} set to zero.
    
    Consider $\mathcal{A}$ on input $z$ and $\hat{\mathcal{A}}$ on inputs $z^{\langle t \rangle}$ for $t\geq 0$.
    We argue by induction that $s^{(t)}(z) = \hat{s}^{(t)}(z^{\langle t \rangle})$.
   The base case $t=-1$ trivially holds.
    The inductive step is
    \begin{align*}
        s^{(t)}(z)
        = u^{(t)}(s^{(t-1)}, z_t)
        = \hat{u}^{(t)}(\hat{s}^{(t-1)}, z^{\langle t \rangle}_{\geq t})
        = \hat{s}^{(t)}(z^{\langle t \rangle}) .
    \end{align*}
    It follows that $\mathcal{A}$ solves $M$:
    \begin{equation*}
       m^{(t)}(s^{(t)}(z))
       = \hat{m}^{(t)}(\hat{s}^{(t)}(z^{\langle t \rangle}), z^{\langle t \rangle}_{> t})
       = (Mz^{\langle t \rangle})_t
       = (Mz)_t .
    \end{equation*}

        Starting from the look-ahead algorithm $\hat{\mathcal{A}}$, we have constructed the algorithm $\mathcal{A}$ without look-ahead.
    It solves the same problem over all inputs, and has the same buffer size. Hence,
    \begin{equation*}
        \widehat{\cost}_M(t; I) \geq \cost_M(t; I).
    \end{equation*}
    Observing that the class of \enquote{no look-ahead algorithms} is a subset of the class of look-ahead algorithms finishes the proof.
\end{proof}
\begin{remark}
 \Cref{lem:future-no-help} does not rely on continuity of the state map.
\end{remark}

\ifdefined\anonymize
\noindent\textbf{Acknowledgments.} We thank Nikita Kalinin for valuable feedback on a preliminary version of this text.
We are also thankful for the STOC 2026 automated pre-submission feedback system for helpful comments.
\fi

\bibliographystyle{alpha}
\bibliography{main}
\end{document}